\documentclass[submission,copyright]{eptcs}
\usepackage[utf8]{inputenc}
\providecommand{\event}{GANDALF 2018} 
\usepackage{breakurl}
\usepackage{underscore}
\usepackage{todonotes}
\usepackage{amsmath}
\usepackage{amsthm}
\usepackage{amssymb}
\usepackage{stmaryrd}
\usepackage{wasysym}
\usepackage{mathrsfs}
\usepackage{cases}

\usepackage[english]{babel}
\usepackage{colortbl}
\usepackage{multirow}
\usepackage{textcomp}
\usepackage{calc}
\usepackage{pgf}
\usepackage{ifthen}
\usepackage{mdwlist}
\usepackage{enumerate}
\usepackage{paralist}
\usepackage{subcaption}

\newtheorem{theorem}{Theorem}[section]

\newtheorem{definition}[theorem]{Definition}

\usepackage{tikz}
\usetikzlibrary{arrows,automata}
\usetikzlibrary{decorations.pathmorphing}
\usetikzlibrary{decorations.markings}
\usetikzlibrary{decorations.shapes}
\usetikzlibrary{matrix}

\DeclareMathAlphabet{\mathpzc}{OT1}{pzc}{m}{it}

\tikzset{
	path/.style={dotted},
	every edge/.style={draw,solid},
	normal/.style={solid},
}

\newcommand*{\optl}{OPTL}
\newcommand*{\nwtl}{NWTL}
\newcommand*{\ltl}{LTL}

\newcommand*{\lnext}{\ocircle}
\newcommand*{\lanext}{\ocircle_{\chi}}
\newcommand*{\lback}{\circleddash}
\newcommand*{\laback}{\circleddash_{\chi}}
\newcommand*{\lguntil}[4]{#3 \mathbin{\mathcal{U}^{#1}_{#2}} #4}
\newcommand*{\luntil}[3]{#2 \mathbin{\mathcal{U}^{#1}} #3}
\newcommand*{\lluntil}[2]{\luntil{}{#1}{#2}}

\newcommand*{\lsince}[3]{#2 \mathbin{\mathcal{S}^{#1}} #3}
\newcommand*{\llsince}[2]{\lsince{}{#1}{#2}}
\newcommand*{\lhyuntil}[2]{\luntil{\uparrow}{#1}{#2}}
\newcommand*{\lhysince}[2]{\lsince{\downarrow}{#1}{#2}}
\newcommand*{\lhtuntil}[2]{\luntil{\downarrow}{#1}{#2}}
\newcommand*{\lhtsince}[2]{\lsince{\uparrow}{#1}{#2}}

\newcommand*{\llglob}{\Box}

\newcommand*{\lapnext}{\lanext^\mathrm{s}}
\newcommand*{\laknext}{\lanext^\mathrm{p}}
\newcommand*{\laenext}{\lanext^\mathrm{end}}
\newcommand*{\laoenext}{\lanext^{\omega \mathrm{end}}}

\newcommand*{\lhypuntil}[2]{\lguntil{\uparrow}{\mathrm{s}}{#1}{#2}}

\newcommand*{\lhyeuntil}[2]{\lguntil{\uparrow}{\mathrm{end}}{#1}{#2}}

\newcommand*{\chain}{\chi}
\newcommand*{\fchain}{\overrightarrow{\chi}}
\newcommand*{\bchain}{\overleftarrow{\chi}}

\newcommand*{\powset}[1]{{\mathscr{P}(#1)}}

\newcommand*{\munw}{\mu}
\newcommand*{\clos}[1]{\operatorname{Cl}({#1})}
\newcommand*{\atoms}[1]{\operatorname{Atoms}({#1})}

\newcommand*{\bfsym}[1]{\boldsymbol{#1}}
\SetSymbolFont{stmry}{bold}{U}{stmry}{m}{n}
\SetSymbolFont{wasy}{bold}{U}{wasy}{m}{n}

\newcommand{\mrk}[1]{{#1}'}
\newcommand{\oldstack}[3]{%
{\ifthenelse{\equal{#1}{1}}{%
\mrk{#2}
}%
{#2}}_{#3}%
}
\newcommand{\stack}[3]{%
[%
{\ifthenelse{\equal{#1}{1}}{%
\mrk{#2}
}%
{#2}}
{\ifthenelse{\equal{#1}{0}}{\ }{} }
{#3}%
]%
}

\newcommand{\tstack}[2]{%
[#1,\ #2]%
}
\newcommand{\tconfig}[3]{\langle #1, \ #2, \ #3 \rangle}

\newcommand{\transition}[1]{\stackrel {{#1}} \vdash}

\newcommand{\va}[1]{\stackrel{#1}{\longrightarrow}}
\newcommand{\ourpath}[1]{\stackrel{#1}{\leadsto}}
\newcommand{\flush}[1]{\stackrel{#1}{\Longrightarrow}}

\newcommand{\ochain}[3]{{}^{#1}\!\left[ #2 \right]\!{}^{#3}}

\newcommand{\symb}[1]{\mathop{smb}(#1)}
\newcommand{\state}[1]{\mathop{st}(#1)}

\title{Temporal Logic and Model Checking for Operator Precedence Languages\footnote{Work partially supported by project AUTOVAM, funded by Fondazione Cariplo and Regione Lombardia.}
}
\author{
  Michele Chiari \quad Dino Mandrioli \quad Matteo Pradella\footnote{Also with IEIIT, Consiglio Nazionale delle Ricerche.}
  \institute{DEIB, Politecnico di Milano \\ P.zza L. Da Vinci, 32, 20133, Milano, Italy}
  \email{michele.chiari@mail.polimi.it  dino.mandrioli@polimi.it matteo.pradella@polimi.it}
}
\def\titlerunning{Operator Precedence Temporal Logic}
\def\authorrunning{M. Chiari, D. Mandrioli \& M. Pradella}
\begin{document}
\maketitle
\begin{abstract}
In the last decades much research effort has been devoted to extending the success of model checking from the traditional field of finite state machines and various versions of temporal logics to suitable subclasses of context-free languages and appropriate extensions of temporal logics. To the best of our knowledge such attempts only covered \emph{structured languages}, i.e. languages whose structure is immediately ``visible'' in their sentences, such as tree-languages or visibly pushdown ones. In this paper we present a new temporal logic suitable to express and automatically verify properties of \emph{operator precedence languages}. This ``historical'' language family has been recently proved to enjoy fundamental algebraic and logic properties that make it suitable for model checking applications yet breaking the barrier of visible-structure languages (in fact the original motivation of its inventor Floyd was just to support efficient \emph{parsing}, i.e. building the ``hidden syntax tree'' of language sentences). We prove that our logic is at least as expressive as analogous logics defined for visible pushdown languages yet covering a much more powerful family; we design a procedure that, given a formula in our logic builds an automaton recognizing the sentences satisfying the formula, whose size is at most exponential in the length of the formula. 

\end{abstract}

\section{Introduction}
Since the pioneering works by Floyd, Hoare, McNaughton, B\"{u}chi and many others, the investigation of the relation between formal language and automata theory and mathematical logic has been an exciting and productive research field, whose main perspective and goal was the \emph{formal correctness verification}, i.e., a mathematical proof that a given design, formalized as a suitable abstract machine, guarantees system requirements, formalized in terms of mathematical logic formulas. Whereas the early work by Floyd, Hoare, Dijkstra and others pursued the full generality of Turing complete computational formalisms, such as normal programming languages, and consequently made the verification problem undecidable and dependent on human inspection and skill, the independent approach by B\"{u}chi, McNaughton and others focused on the restricted but practically quite relevant families of finite state machines (FSMs) on the one side and of \emph{monadic logics} on the corresponding side. The main achievements on this respect have been the characterization of regular languages --those recognized by FSMs-- in terms of monadic second order logic (MSO)~\cite{bib:Buchi1960a} and the definition of an incredible number of subfamilies that are all equivalent between each other and are characterized in terms of monadic first-order logic (MFO)~\cite{McNaughtonPapert71}.

Such foundational results, however, remained of essentially theoretic interest because the formal correctness problem, though decidable, remains of intractable complexity for MSO and MFO logics (see e.g., \cite{DBLP:FrickG04}).
The state of the art, however, had a dramatic breakthrough with the advent of \emph{temporal logic and model checking} \cite{Emerson90}: for now classic logics such as \emph{linear time temporal logic} (LTL), CTL* and others,
the 
model checking problem, though PSPACE complete, has a time complexity bounded by ``only'' a singly exponential function of formula's length \footnote{The parallel field of correctness verification for Turing complete formalisms, instead, ignited many efforts on general purpose ``semiautomatic" theorem proving.}.

Not surprisingly, the success of model checking based on FSMs and several extensions thereof, e.g., their timed version \cite{AlurDill1994a}, generated the wish of extending it to the case of context-free languages (CFLs) to serve much larger application fields such as general purpose programming languages, large web data based on XML, HTML,  etc. Whereas the case of the full CFL family is made difficult if not impossible due to the lack of fundamental decidability and closure properties, some early results have been obtained in the context of \emph{structured CFLs}: with this term we mean languages whose typical tree-shaped structure is immediately visible in their sentences; the first instance of such language families are parenthesis languages \cite{McNaughton67}, which correspond to regular tree-languages \cite{Tha67}; among various extensions thereof special attention have received \emph{input-driven} (ID) \cite{Input-driven}, alias \emph{visibly push-down} languages (VPLs) \cite{jacm/AlurM09}. Thanks to the fact that they enjoy many of the fundamental closure properties of regular languages, VPLs too have been characterized in terms of a suitable MSO logic \cite{jacm/AlurM09}; early attempts have also been done to support their model checking by exploiting suitable extensions of temporal logics --whether linear time \cite{lmcs/AlurABEIL08} or branching time \cite{DBLP:journals/toplas/AlurCM11}. It is also worth mentioning some parallel results on model checking state machines whose nondeterministic computations have a typical tree-shaped structure by means of a variant of alternation-free modal mu-calculus (a branching time logic) \cite{BurkartSteffen1992a}, whose relationship with pushdown games is discussed in \cite{Walukiewicz2001}, and with both linear and branching time logic in \cite{BouajjaniEM97}.

In this paper we address the same problem in the context of a much larger subfamily of CFL, i.e., \emph{operator precedence languages} (OPLs). OPLs have been invented by R. Floyd to support efficient deterministic parsing of programming languages \cite{Floyd1963}; subsequently, we showed that they enjoy many of the algebraic properties of structured CFLs \cite{Crespi-ReghizziMM1978}; after several decades we resumed the investigation of this family by envisioning their application to various practical state of the art problems, noticeably automatic verification and model checking.\footnote{We also devised efficient parsers for OPLs by exploiting parallelism \cite{BarenghiEtAl2015}, but this issue is not the object of the present paper.} We have shown that OPLs strongly generalize VPLs \cite{CrespiMandrioli12} not only in terms of strict set theoretic inclusion but in that they allow to describe typical programming language constructs such as traditional arithmetic expressions whose structure is not immediately ``visible'' in their sentences, unless one does not take into account the \emph{implicit precedence} of, e.g., multiplicative operators over additive ones. We have built an MSO characterization of OPLs \cite{LonatiEtAl2015} by extending in a non-trivial way the key relation between ``matching string positions'' introduced in \cite{Lautemann94} and exploited in \cite{jacm/AlurM09} for the logical characterization of VPLs. All in all OPLs appear to enjoy many if not all of the pleasant algebraic and logic properties of structured CFLs   but widen their application field in a dramatic way. For a more comprehensive description of OPL properties and their relations with other CFL subfamilies see \cite{MP18}.

Here we move one further step in the path toward building model checking algorithms for OPLs --and the wide application field that they can support-- of comparable efficiency with other state of the art tools based on the FSM formalism. After resuming the basic background to make the paper self-contained (Section~\ref{sec:background}), in Section~\ref{sec:OPTL} we introduce our \emph{operator precedence temporal logic} (OPTL) which is inspired by temporal logics for nested words, in particular, the logic NWTL\cite{lmcs/AlurABEIL08} but requires many more technicalities due to the lack of the ``matching relation'' \cite{Lautemann94} typical of parenthesis languages; we formally define OPTL syntax and semantics and provide examples of its usage and its generality. In Section~\ref{sec:OPTL vs NWTL} we show that OPTL defines a larger language family than \cite{lmcs/AlurABEIL08}'s NWTL: 
every NWTL formula can be automatically translated in linear time into an equivalent OPTL formula; strict inclusion follows from the language family inclusion; again, we emphasize that such an inclusion is not just a set theoretical property but shows a much wider application field. Then, in Section~\ref{sec:model checking} we provide the theoretical basis for model checking OP automata (OPAs) against OPTL formulas, i.e., an algorithm which, for any given OPTL formula of length $n$, builds an equivalent nondeterministic OPA of size $2^n$, i.e., the same size of analogous constructions for less powerful automata and less expressive logics.
For the sake of brevity in this paper we focus mainly on finite length languages; at the end of the section, however, we provide a few hints showing how our model checking  OPAs can be extended to deal with $\omega$-languages.
Section~\ref{sec:conclusion} concludes and envisages several further research steps.

\section{Operator Precedence Languages and Automata} \label{sec:background}

Operator Precedence languages are normally defined through their generating grammars
\cite{Floyd1963}; in this paper, however, we
characterize them through their accepting automata
\cite{LonatiEtAl2015} which are the natural way to state equivalence
properties with logic characterization. We assume some
familiarity with classical language theory concepts such as
context-free grammar, parsing, shift-reduce algorithm, syntax tree~\cite{GruneJacobs:08}.

Let $\Sigma = \{a_1, \dots, a_n \}$ be an alphabet. The empty string is denoted
$\epsilon$. 
We use a special symbol \# not in $\Sigma$ to mark the beginning and
the end of any string. This is consistent with the operator
parsing technique, which requires the look-back and look-ahead of one
character to determine the next action \cite{GruneJacobs:08}.

\begin{definition}\label{def:opm}
  An \textit{operator precedence matrix} (OPM) $M$ over an alphabet
  $\Sigma$ is a partial function $(\Sigma \cup \{\#\})^2 \to \{\lessdot,
  \doteq, \gtrdot\}$, that with each ordered pair $(a,b)$
  associates the OP relation $M_{a,b}$ holding between $a$ and $b$;  if the function is total we say that M is \emph{complete}.  We
  call the pair $(\Sigma, M)$ an \emph{operator precedence alphabet}.  Relations $\lessdot, \doteq, \gtrdot$, are respectively named 
  \emph{yields precedence, equal in precedence}, and \emph{takes precedence}.
By convention, the initial \# can only yield precedence, and other
symbols can only take precedence on the ending \#.
If $M_{a,b} = \circ$, where $\circ \in \{\lessdot, \doteq, \gtrdot \}$,
we write $a \circ b$.  For $u,v \in \Sigma^+$ we write $u \circ v$ if
$u = xa$ and $v = by$ with $a \circ b$.
\end{definition}
%


\begin{definition}\label{def:OPA}
An  \emph{operator precedence automaton (OPA)} is a tuple
$\mathcal A = (\Sigma, M, Q, I, F, \delta) $ where:
\begin{itemize}
\item $(\Sigma, M)$ is an operator precedence alphabet,
\item $Q$ is a set of states (disjoint from $\Sigma$),
\item $I \subseteq Q$ is the set of initial states,
\item $F \subseteq Q$ is the set of final states,
\item $\delta \subseteq Q \times ( \Sigma \cup Q) \times Q$ is the transition relation, which is the union of three disjoint relations:
\[
\delta_{\text{shift}}\subseteq Q \times \Sigma \times Q,
\quad 
\delta_{\text{push}}\subseteq Q \times \Sigma \times Q,
\quad 
\delta_{\text{pop}}\subseteq Q \times Q \times Q.
\]
\end{itemize}
An OPA is deterministic iff
$I$ is a singleton, and
all three components of $\delta$ are --possibly partial-- functions:
$\delta_{\text{shift}}: Q \times \Sigma \to Q,
\  
\delta_{\text{push}}: Q \times \Sigma \to Q,
\
\delta_{\text{pop}}: Q \times Q \to Q.$

\end{definition}

To define the semantics of the automaton, we need some new notations.
We use letters $p, q, p_i, q_i, \dots $ to denote states in $Q$.
We will sometimes use $q_0 \va{a}{q_1}$ for $(q_0, a, q_1) \in  \delta_{\text{shift}} \cup \delta_{\text{push}}$,
 $q_0 \flush{q_2}{q_1}$  for $(q_0, q_2, q_1) \in  \delta_{\text{pop}}$,
and ${q_0} \ourpath{w} {q_1}$, if the automaton can read $w \in \Sigma^*$ going from $q_0$ to $q_1$.
Let  $\Gamma$ be	$\Sigma \times Q$ and let $\Gamma' = \Gamma \cup  \{\bot\} $ be the \textit{stack alphabet}; 
we denote symbols in $\Gamma'$ as $\tstack aq$ or $\bot$.
We set $\symb {\tstack aq} = a$, $\symb {\bot}=\#$, and
$\state {\tstack aq} = q$.
Given a stack content $\Pi =  \pi_n \dots  \pi_2  \pi_1 \bot$, with $ \pi_i \in \Gamma$ , $n \geq 0$, 
we set $\symb \Pi = \symb{\pi_n}$ if $n \geq 1$, $\symb \Pi = \#$ if $n = 0$.

A \emph{configuration} of an OPA is a triple $c = \tconfig w q \Pi$,
where $w \in \Sigma^*\#$, $q \in Q$, and $\Pi \in \Gamma^*\bot$.

A \emph{computation} or \emph{run} of the automaton is a finite sequence
$c_0 \transition{} c_1 \transition{} \dots \transition{} c_n$
of \emph{moves} or \emph{transitions} 
$c_i \transition{} c_{i+1}$; 
there are three kinds of moves, depending on the precedence relation between the symbol on top of the stack and the next symbol to read:

\smallskip
\noindent {\bf push move:} if $\symb \Pi \lessdot \ a$ then
$
\tconfig {ax} p  \Pi\transition{} \tconfig {x} q {\tstack   a p\Pi }$, with $(p,a, q) \in \delta_{\text{push}}$;

\smallskip
\noindent {\bf shift move:} if $a \doteq b$ then 
$
\tconfig {bx} q { \tstack a p \Pi}  \transition{} \tconfig x  r { \tstack b p \Pi}$, with $(q,b,r) \in \delta_{\text{shift}}$;

\smallskip

\noindent {\bf pop move:} if $a \gtrdot b$
then 
$
\tconfig {bx} q  { \tstack a p \Pi}\transition{} \tconfig {bx} r \Pi $, with $(q, p, r) \in \delta_{\text{pop}}$.

Observe that shift and pop moves are never performed when the stack contains only $\bot$.

Push and shift moves update the current state of the automaton according to the transition relations $\delta_{\text{push}}$ and  $\delta_{\text{shift}}$, respectively: push moves put a new element on top of the stack consisting of the input symbol together with the current state of the automaton, whereas shift moves update the top element of the stack by \textit{changing its input symbol only}.
Pop moves remove the element on top of the stack,
and update the state of the automaton according to $\delta_{\text{pop}}$ on the basis of the pair of states consisting of the current state of the automaton and the state of the removed stack symbol;
pop moves do not consume the input symbol, which is used only to establish the $\gtrdot$ relation, remaining available for the next move.

\noindent The automaton accepts the language
$
L(\mathcal A) = \left\{ x \in \Sigma^* \mid  \tconfig {x\#} {q_I} {\bot} \vdash ^* 
\tconfig {\#} {q_F}{\bot} , \allowbreak q_I \in I, \allowbreak q_F \in F \right\}.
$

\begin{definition}\label{def:chain}
A \emph{simple chain} is a string $c_0 c_1 c_2 \dots c_\ell c_{\ell+1}$,
written as
$
\ochain {c_0} {c_1 c_2 \dots c_\ell} {c_{\ell+1}},
$
such that:
$c_0, \allowbreak c_{\ell+1} \in \Sigma \cup \{\#\}$,
$c_i \in \Sigma$ for every $i = 1,2, \dots \ell$ ($\ell \geq 1$),
and $c_0 \lessdot c_1 \doteq c_2 \dots c_{\ell-1} \doteq c_\ell
\gtrdot c_{\ell+1}$.

A \emph{composed chain} is a string 
$c_0 s_0 c_1 s_1 c_2  \dots c_\ell s_\ell c_{\ell+1}$, 
where
$\ochain {c_0}{c_1 c_2 \dots c_\ell}{c_{\ell+1}}$ is a simple chain, and
$s_i \in \Sigma^*$ is the empty string 
or is such that $\ochain {c_i} {s_i} {c_{i+1}}$ is a chain (simple or composed),
for every $i = 0,1, \dots, \ell$ ($\ell \geq 1$). 
Such a composed chain will be written as
$\ochain {c_0} {s_0 c_1 s_1 c_2 \dots c_\ell s_\ell} {c_{\ell+1}}$.

The pair made of the first and the last symbols of a chain is called its {\em context}.
\end{definition}

\begin{definition}
  A finite word $w$ over $\Sigma$ is \emph{compatible} with an OPM $M$ iff
  for each pair of letters $c, d$, consecutive in $w$, $M_{cd}$ is defined and,
  for each substring $x$ of $\# w \#$ which is a chain of the form $^a[y]^b$,
  $M_{a b}$ is defined.
\end{definition}

E.g., the word 
$\mathbf{call} \  
 \mathbf{handle} \ 
 \mathbf{call} \ 
 \mathbf{call} \ 
 \mathbf{call} \ 
 \mathbf{throw} \ 
 \mathbf{throw} \ 
 \mathbf{throw} \ 
 \mathbf{ret}$ 
of Figure~\ref{fig:op-matrix-mcall} is compatible with $M_\mathbf{call}$. 
In the same figure all the resulting chains are reported, e.g. 
$\ochain {\mathbf{call}} {\mathbf{call}} {\mathbf{throw}}$,
$\ochain {\mathbf{handle}} {\mathbf{throw}} {\mathbf{ret}}$ 
are simple chains,
while 
$\ochain {\mathbf{call}} {\mathbf{call} [ \mathbf{call} ] }  {\mathbf{throw}}$, 
$\ochain {\mathbf{handle}} {[ [ [ \mathbf{call} [ \mathbf{call} [ \mathbf{call} ] ] ] \mathbf{throw} ] \mathbf{throw} ] \mathbf{throw} } {\mathbf{ret}}$
are composed chains.

\begin{definition}
Let $\mathcal A$ be an OPA.
We call a \emph{support} for the simple chain
$\ochain {c_0} {c_1 c_2 \dots c_\ell} {c_{\ell+1}}$
any path in $\mathcal A$ of the form
$q_0
\va{c_1}{q_1}
\va{}{}
\dots
\va{}q_{\ell-1}
\va{c_{\ell}}{q_\ell}
\flush{q_0} {q_{\ell+1}}$, where the arrow labeled $c_1$ corresponds to a push move whereas the remaining ones denote shift moves. 
The label of the last (and only)  pop is exactly $q_0$, i.e. the first state of the path; this pop is executed because of relation $c_\ell \gtrdot c_{\ell+1}$.

\noindent We call a \emph{support for the composed chain} 
$\ochain {c_0} {s_0 c_1 s_1 c_2 \dots c_\ell s_\ell} {c_{\ell+1}}$
any path in $\mathcal A$ of the form
\begin{equation}
\label{eq:compchain}
q_0
\ourpath{s_0}{q'_0}
\va{c_1}{q_1}
\ourpath{s_1}{q'_1}
\va{c_2}{}
\dots
\va{c_\ell} {q_\ell}
\ourpath{s_\ell}{q'_\ell}
\flush{q'_0}{q_{\ell+1}}
\end{equation}
where, for every $i = 0, 1, \dots, \ell$: 
if $s_i \neq \epsilon$, then $q_i \ourpath{s_i}{q'_i} $ 
is a support for the chain $\ochain {c_i} {s_i} {c_{i+1}}$, else $q'_i = q_i$.

\end{definition}

The chains fully determine the structure of the parsing of any
automaton over $(\Sigma, M)$. If the automaton performs the computation
$
\langle sb, q_i, [a, q_j] \Pi \rangle \vdash^*
\langle b,  q_k, \Pi \rangle
$
then $\ochain asb$ 
is necessarily a chain over $(\Sigma, \allowbreak M)$ and there exists a support
like \eqref{eq:compchain} with $s = s_0 c_1 \dots c_\ell s_\ell$ and $q_{\ell+1} = q_k$.
The above computation corresponds to the parsing by the
automaton of the string $s_0 c_1 \dots c_\ell s_{\ell}$ within the
context $a$,$b$. Such context contains all
information needed to build the subtree whose frontier is that string.
This is a distinguishing feature of OP languages: we call it the \emph{locality principle}.

\begin{definition}\label{def:maxfa}
Given $(\Sigma, M)$,  let us consider the OPA
$\mathcal A(\Sigma, M)$ $=$ $\langle \Sigma, M,$ $\{q\}, \{q\}, \{q\}, \delta_{max}
\rangle $ where  $\delta_{max}(q,q) = q$, and $\delta_{max}(q,c) = q$,
$\forall c \in \Sigma$.
We call $\mathcal A(\Sigma, M)$ the \emph{OP Max-Automaton} over $\Sigma, M$.
\end{definition}

For a max-automaton $\mathcal A(\Sigma,M)$ each chain has a support;
since there is a chain $\ochain{\#}{s}{\#}$ for any string $s$
compatible with $M$, a string is accepted by $\mathcal A(\Sigma, M)$
iff it is compatible with $M$.  Also, whenever $M$ is complete, each
string is compatible with $M$, hence accepted by the max-automaton; 
thus, when M
is complete the max-automaton defines the
universal language $\Sigma^*$ by assigning to any string the (unique)
structure compatible with the OPM.
Considering $M_{\mathbf{call}}$ of Figure~\ref{fig:op-matrix-mcall},
if we take e.g. the string  
$
\mathbf{ret} \
\mathbf{call} \  
\mathbf{handle}
$, it is accepted by the max-automaton and its structure is 
$
\#[[\mathbf{ret}]
\mathbf{call}
[\mathbf{handle}]]
\#.
$

\begin{figure}
\begin{subfigure}{0.3\linewidth}
\centering
\footnotesize{
\[
\begin{array}{r | c c c c}
                & \mathbf{call} & \mathbf{ret} & \mathbf{handle} & \mathbf{throw} \\
\hline
\mathbf{call}   & \lessdot      & \doteq       & \lessdot        & \gtrdot        \\
\mathbf{ret}    & \gtrdot       & \gtrdot      & \lessdot        & \gtrdot        \\
\mathbf{handle} & \lessdot      & \gtrdot      & \lessdot        & \lessdot       \\
\mathbf{throw}  & \gtrdot       & \gtrdot      & \gtrdot         & \gtrdot        \\
\end{array}
\]}
\end{subfigure}
\begin{subfigure}{0.8\linewidth}
\[
\# [ \mathbf{call} [ \mathbf{handle} 
[ [ [ [ \mathbf{call} [ \mathbf{call} [ \mathbf{call} ] ] ] 
\mathbf{throw} ] \mathbf{throw} ] \mathbf{throw} ] ] \mathbf{ret} ] \#
\]
\end{subfigure}
\caption{OPM $M_\mathbf{call}$ and an example string with all the chains evidenced by brackets.}
\label{fig:op-matrix-mcall}
\end{figure}

In conclusion, given an OP alphabet, the OPM $M$ assigns a structure
to any compatible string in $\Sigma^*$; unlike parentheses languages such a structure is 
not visible in the string, and must be built by means of a non-trivial parsing 
algorithm.  An OPA defined on the
OP alphabet selects an appropriate subset within the
``universe'' of strings compatible with $M$. In some sense this property is yet 
another variation of
the fundamental Chomsky-Sh\"utzenberger theorem.
All above definitions are extended to the case of infinite strings in the traditional way~\cite{LonatiEtAl2015}; for this reason and for the sake of brevity we will deal with the $\omega$-case only in Section \ref{MC infinite} where the application of novel techniques is necessary.
For a more complete description of the OPL family and of its relations with other CFLs we refer the reader to \cite{MP18}.

\section{Operator Precedence Temporal Logic} \label{sec:OPTL}

Languages recognized by different OPAs on a given OP alphabet form a Boolean algebra
\cite{Crespi-ReghizziMM1978}, i.e. they are closed under union, intersection and complement:
these properties allow us to define \optl{}
(\emph{Operator Precedence Temporal Logic}).
Given an OP alphabet, each well-formed \optl{} formula characterizes a subset of the
universal language based on that alphabet.
Due to the closure properties above, for each \optl{} formula it is possible to
identify an OPA that recognizes the same language denoted by it,
opening the way for model checking of \optl{}.

Next, we present the syntax of \optl{}
explaining its meaning by means of simple examples, then formally define its semantics.
We explore the relationships between the \optl{} operators in Section~\ref{subsec:optl-equivalence}.

\subsection{Syntax and Semantics}
\label{subsec:optl-syntax}

\optl{} is a propositional temporal logic:
the truth of formulas depends on the atomic propositions holding in each word position.
Let $AP$ be the finite set of atomic propositions:
the semantics of \optl{} relies on an OP alphabet based on $\powset{AP}$,
so terminal characters are subsets of $AP$,
and word structure is given by an OPM defined on $\powset{AP}$.
Since defining an OPM on a power set is often uselessly cumbersome,
in this paper we define it on a set $\Sigma \subseteq AP$,
whose elements are called ``structural labels'', and typeset in bold.
The corresponding OPM on $\powset{AP}$ can be derived from the structural label
contained in each subset of $AP$,
leaving the matrix undefined for subsets not containing exactly one element of $\Sigma$. For example, in the word of Figure~\ref{fig:optl-example-word},
pos.~1 is labeled with the set $\{\mathbf{call}, \mathrm{p}_a\}$,
and pos.~3 with $\{\mathbf{call}, \mathrm{p}_b\}$:
they both contain $\mathbf{call}$,
so they are in the $\lessdot$ relation according to matrix $M_\mathbf{call}$
of Figure~\ref{fig:op-matrix-mcall}.

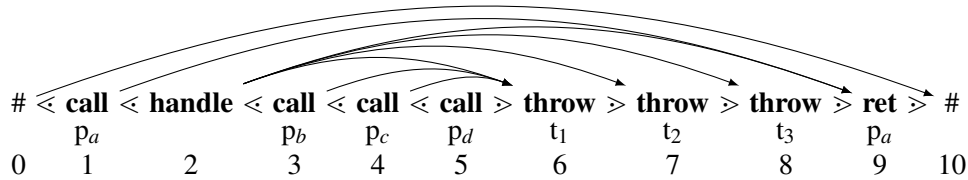
\begin{figure}
\centering
\begin{tikzpicture}
  [edge/.style={->, >=latex}]
\matrix (m) [matrix of math nodes, column sep=-4, row sep=-4]
{
  \bfsym{\#}
  & \lessdot & \mathbf{call}
  & \lessdot & \mathbf{handle}
  & \lessdot & \mathbf{call}
  & \lessdot & \mathbf{call}
  & \lessdot & \mathbf{call}
  & \gtrdot & \mathbf{throw}
  & \gtrdot & \mathbf{throw}
  & \gtrdot & \mathbf{throw}
  & \gtrdot & \mathbf{ret}
  & \gtrdot & \bfsym{\#} \\
  & & \mathrm{p}_a
  & &
  & & \mathrm{p}_b
  & & \mathrm{p}_c
  & & \mathrm{p}_d
  & & \mathrm{t}_1
  & & \mathrm{t}_2
  & & \mathrm{t}_3
  & & \mathrm{p}_a
  & & \\
  0
  & & 1
  & & 2
  & & 3
  & & 4
  & & 5
  & & 6
  & & 7
  & & 8
  & & 9
  & & 10 \\
};
\draw[edge] (m-1-1) to [out=20, in=160] (m-1-21);
\draw[edge] (m-1-3) to [out=20, in=160] (m-1-19);
\draw[edge] (m-1-5) to [out=20, in=160] (m-1-19);
\draw[edge] (m-1-5) to [out=20, in=160] (m-1-17);
\draw[edge] (m-1-5) to [out=20, in=160] (m-1-15);
\draw[edge] (m-1-5) to [out=20, in=160] (m-1-13);
\draw[edge] (m-1-7) to [out=20, in=160] (m-1-13);
\draw[edge] (m-1-9) to [out=20, in=160] (m-1-13);
\end{tikzpicture}
\caption{An example of execution trace, according to $M_\mathbf{call}$
  (Figure~\ref{fig:op-matrix-mcall}).
  Chains are highlighted by arrows joining their context;
  structural labels are typeset in bold,
  while other atomic propositions are shown below them.
  First, procedure $\mathrm{p}_a$ is called (pos.~1),
  and it installs an exception handler in pos.~2.
  Then, three nested procedures are called,
  and the innermost one ($\mathrm{p}_d$) throws a sequence of exceptions,
  which are all caught by the handler.
  Finally, $\mathrm{p}_a$ returns, uninstalling the handler.}
\label{fig:optl-example-word}
\end{figure}

\noindent {\bf Syntax. } The syntax of \optl{} is based on the following grammar,
where $\mathrm{a}$ denotes any symbol in $AP$:
\begin{align*}
\varphi := \;
& \mathrm{a}
\mid \neg \varphi
\mid (\varphi \land \varphi)
\mid \lnext{\varphi}
\mid \lanext{\varphi}
\mid \lback{\varphi}
\mid \laback{\varphi} \mid\\ 
& (\lluntil{\varphi}{\varphi})
\mid (\luntil{\boxdot}{\varphi}{\varphi})
\mid (\llsince{\varphi}{\varphi})
\mid (\lsince{\boxdot}{\varphi}{\varphi}) 
\mid (\luntil{\varobar}{\varphi}{\varphi})
\mid (\lsince{\varobar}{\varphi}{\varphi})
\end{align*}
We informally show the meaning of \optl{} operators
by referring to the word of Figure~\ref{fig:optl-example-word},
w.r.t the OPM of Fig.~\ref{fig:op-matrix-mcall}.
The $\lnext$ and $\lback$ symbols denote the next and back operators from \ltl{},
while the undecorated $\lluntil{}{}$ and $\llsince{}{}$
operators are \ltl{} until and since.
The $\lanext$ and $\laback$ operators,
which we call \emph{matching next} and \emph{matching back},
express properties on string positions in the chain relation
(which will be formally defined later on) with the current one.
For example, formula $\lanext \mathbf{throw}$ is true in positions containing
a $\mathbf{call}$ to a procedure that is terminated by an exception
thrown by an inner procedure, such as 3 and 4 of Fig.~\ref{fig:optl-example-word},
because pos.\ 3 forms a chain with pos.\ 6, in which $\mathbf{throw}$ holds, and so on.
Formula $\laback \mathbf{handle}$ is true in handled $\mathbf{throw}$ positions,
such as 6, 7 and 8, because e.g.\ pos. 2 forms a chain with 6, and $\mathbf{handle}$ holds in 2.
The $\luntil{\boxdot}{}{}$ and $\lsince{\boxdot}{}{}$ operators,
called \emph{operator precedence summary until and since},
are inspired to the homonymous $\luntil{\sigma}{}{}$ and $\lsince{\sigma}{}{}$
operators from \cite{lmcs/AlurABEIL08},
and are path operators that can ``jump'' over chain bodies;
the symbol $\boxdot$ is a placeholder for one or more precedence relations allowed in the path
(e.g. $\luntil{\lessdot \doteq}{}{}$ or $\luntil{\gtrdot \lessdot}{}{}$ and so on).
Formula $\luntil{\gtrdot}{(\mathbf{call} \lor \mathbf{throw})}{\mathbf{ret}}$
is true in pos.\ 3 because there is a path that jumps over the chain between 3 and 6,
and goes on with positions 7, 8 and 9, which are in the $\gtrdot$ relation:
pos.\ 3 and from 6 to 8 satisfy $\mathbf{call} \lor \mathbf{throw}$,
while pos.\ 9 satisfies $\mathbf{ret}$.
In $\luntil{\varobar}{}{}$ and $\lsince{\varobar}{}{}$,
$\varobar$ is a placeholder for $\uparrow$ or $\downarrow$;
these peculiar path operators are called \emph{hierarchical until and since},
and they express properties about the multiple positions
in the chain relation with the current one:
their associated paths can dive up and down between such positions.
For example, $\lhyuntil{\mathbf{throw}}{\mathrm{t}_3}$
and $\lhysince{\mathbf{throw}}{\mathrm{t}_1}$ hold in pos.\ 2,
because there is path 6-7-8 made of ending positions of chains starting in 2,
such that $\mathbf{throw}$ holds until $\mathrm{t}_3$ holds
(or $\mathbf{throw}$ has held since $\mathrm{t}_1$ held).
Formulas $\lhtuntil{\mathbf{call}}{\mathrm{p}_c}$
and $\lhtsince{\mathbf{call}}{\mathrm{p}_b}$ hold in pos.\ 6,
because of path 3-4, made of positions where a chain ending in 6 starts,
and whose labels satisfy the appropriate until and since conditions.

Many relevant properties can be expressed in \optl{}:
formula $\llglob [\mathbf{handle} \implies \lanext \mathbf{ret}]$,
where $\llglob \psi$ is a shortcut for $\neg (\lluntil{\top}{\neg \psi})$,
holds if all exception handlers are properly uninstalled by a return statement.
Formula $\llglob [\mathbf{throw} \implies \neg (\lhtsince{\top}{\mathrm{p}_b})]$
is false if procedure $\mathrm{p}_b$ is terminated by an exception;
and formula $\neg (\lhyuntil{\top}{(\mathbf{throw} \land \lhtuntil{\top}{\mathbf{call}})})$
is true in $\mathbf{handle}$s catching only throw statements not interrupting any procedure.
These properties may not be expressible in \nwtl{},
because its nesting relation is one-to-one, and fails to model situations in which
a single entity is in relation with multiple other entities.
\smallskip

\noindent {\bf Semantics. } The \optl{} semantics is based on the OP word structure
$\langle U, M_\powset{AP}, P \rangle$
where
\begin{itemize}
\item
  $U = \{0, 1, \dots, n, n+1\}$, with $n \in \mathbb{N}$ is a set of word positions;
\item
  $M_\powset{AP}$ is an operator precedence matrix on $\powset{AP}$;
\item
  $P \colon U \to \powset{AP}$ is a function associating each word position in $U$
  with the set of atomic propositions that hold in that position,
  with $P(0) = P(n+1) = \{\bfsym{\#}\}$.
\end{itemize}

The word structure is given by the OPM $M_\powset{AP}$,
which associates a precedence relation to each pair of positions,
based on the subset of $AP$ associated to them by $P$.
For OPTL to be able to denote the universal language $\powset{AP}^*$,
$M_\powset{AP}$ must be complete.
In the following we will denote subsets of $AP$ by lowercase letters in italic,
such as $a \in \powset{AP}$.
For any $i, j \in U$ we write, $i \lessdot j$
if $a = P(i)$, $b = P(j)$ and the relation  $a \lessdot b$ is in $M_\powset{AP}$.

The semantics of \optl{} deeply relies on the concept of \emph{chain},
presented in Section~\ref{sec:background}.
We define the chain relation $\chain \subseteq U \times U$
so that $\chain(i, j)$ holds between two positions $i < j$
if $i$ and $j$ form the context of a chain.
In case of composed chains, this relation is not one-to-one:
there may be positions where multiple chains start or end.
Therefore, we additionally define two one-to-one relations,
helpful in identifying the largest chain starting or ending in a word position.
More formally, the \emph{maximal forward} chain relation is defined so that
$\fchain(i,j) \iff j = \max\{k \in U \mid \chain(i,k) \}$ for any $i, j \in U$,
and the \emph{maximal backward} chain is defined as
$\bchain(i,j) \iff i = \min\{k \in U \mid \chain(k,j) \}$.
The max. forward (resp. backward) chain relation can be undefined
for a pair of positions if either they are the context of no chain,
or if they are the context of a chain which is not forward- (resp. backward-) maximal.

Let $w$ be an OP word, and $\mathrm{a} \in AP$.
Then, for any position $i \in U$ of $w$, we have $(w, i) \models \mathrm{a}$
if $\mathrm{a} \in P(i)$.
Operators such as $\land$ and $\neg$ have the usual semantics from propositional logic,
while $\lnext$ and $\lback$ have the same semantics as in \ltl{}
(i.e. $(w, i) \models \lnext \varphi$ iff $(w, i+1) \models \varphi$,
and similarly for $\lback$).

The $\lanext$ and $\laback$ operators express properties regarding the positions that form a maximal chain that starts (resp. ends) in the current one: $(w,i) \models \lanext{\varphi}$ iff there exists a position
$j \in U$ such that $\fchain(i,j)$ and $(w,j) \models \varphi$;
symmetrically, $(w,i) \models \laback{\varphi}$ iff there exists a position
$j \in U$ such that $\bchain(j,i)$ and $(w,j) \models \varphi$.
In Figure~\ref{fig:optl-example-word},
$(w, 3) \models \lanext \mathbf{throw}$
because $\fchain(3, 6)$ and $(w, 6) \models \mathbf{throw}$;
$(w, 6) \models \laback \mathbf{handle}$ holds because $\bchain(2, 6)$
and $(w, 2) \models \mathbf{handle}$,
but $(w, 6) \not\models \laback \mathrm{p}_b$ because chain $\chain(3,6)$ is not backward-maximal
(although it is forward-maximal).

A \emph{path} of length $n \in \mathbb{N}$ between
$i, j \in U$ is a sequence of positions
$i_1 < i_2 < \dots < i_n$, with $i \leq i_1$ and $i_n \leq j$.
The \emph{until} operator on a set of paths $\Pi$ is defined as follows:
for any word $w$ and position $i \in U$,
and for any two \optl{} formulas $\varphi$ and $\psi$,
$(w, i) \models \luntil{\Pi}{\varphi}{\psi}$
iff there exist a position $j \in U$, $j \geq i$,
and a path $i_1 < i_2 < \dots < i_n$ between $i$ and $j$ in $\Pi$
such that $(w, i_k) \models \varphi$ for any $1 \leq k < n$,
and $(w, i_n) \models \psi$.
The \emph{since} operator is defined symmetrically.
Note that a path from $i$ to $j$ does not necessarily start in $i$ and end in $j$,
but it may do in positions between them.
However, this will only happen with hierarchical paths.
We define the different kinds of until/since operators by associating them with suitable set of paths.

The \emph{linear until} ($\lluntil{\varphi}{\psi}$)
and \emph{since} ($\llsince{\varphi}{\psi}$) operators,
based on linear paths, have the same semantics as in \ltl{}.
A \emph{linear path} starting in position $i \in U$
is such that $i_1 = i$ and for any $1 \leq k < n$ $i_{k+1} = i_k + 1$.

The \emph{OP-summary until} operator exploits the $\fchain$ relation
to express properties on paths that skip chain bodies,
also keeping precedence relations between consecutive word positions into account.
\begin{definition} \label{def:forward-opsp}
Given a set $O \subseteq \{ \lessdot, \doteq, \gtrdot \}$,
the $\luntil{O}{}{}$ operator is based on the class of
\emph{forward OP-summary} paths.
A path of this class between $i$ and $j \in U$
is a sequence of positions $i = i_1 < i_2 < \dots < i_n = j$ such that,
for any $1 \leq k < n$,
\[
i_{k+1} =
\begin{cases}
  h &\text{if $\fchain(i_k,h)$ and $h \leq j$}; \\
  i_{k} + 1 &\text{if $i_k \odot i_k + 1$ with $\odot \in O$, otherwise.}
\end{cases}
\]
\end{definition}

There exists at most one forward OP-summary path between any two positions.
For example, in Figure~\ref{fig:optl-example-word}, if we take $O = \{ \gtrdot \}$ as in
$\luntil{\gtrdot}{(\mathbf{call} \lor \mathbf{throw})}{\mathbf{ret}}$,
the path between 3 and 9 is made of pos.\ 3-6-7-8-9,
because $\fchain(3,6)$ and the body of this chain is skipped, and $6 \gtrdot 7$,
$7 \gtrdot 8$ and $8 \gtrdot 9$.
If we took e.g. $O = \{ \doteq, \lessdot \}$, there would be no such path,
because consecutive positions in the $\gtrdot$ relation are not considered.
With $O = \{ \gtrdot, \lessdot \}$, the path between 2 and 6
does not skip the body of chain $\chain(2,6)$, because it is not forward-maximal:
it is the linear path 2-3-4-5-6.
The \emph{OP-summary since} operator is based on \emph{backward OP-summary} paths,
which are symmetric to their until counterparts,
 relying on the $\bchain$ relation instead of $\fchain$.
For example, $\lsince{\lessdot}{(\mathbf{throw} \lor \mathbf{handle})}{\mathbf{call}}$
holds in pos.~8 because of path 1-2-8, that skips the body of chain $\bchain(2,8)$
and satisfies $\mathbf{throw} \lor \mathbf{handle}$ in 2 and 8,
and $\mathbf{call}$ in 1.
Again, bodies of chains that are not backward-maximal cannot be skipped.

While summary operators are only aware of the maximal chain relations,
hierarchical operators can express properties discriminating between all other chains.
The \emph{hierarchical yield-precedence} until and since operators,
denoted as $\lhyuntil{}{}$ and $\lhysince{}{}$ respectively,
are based on paths made of the ending positions
of non-maximal chains starting in the current position $i \in U$.
One of such paths is a sequence of word positions $i_1 < i_2 < \dots < i_n$, with $i < i_1$,
such that for any $1 \leq k \leq n$ we have $i \lessdot i_k$
and $\chain(i,i_k)$, and, additionally, there is no $i'_k$ that satisfies
these two properties and $i_{k-1} < i'_k < i_k$.
Moreover, for the until operator $i_1$ must be the leftmost position enjoying
the above properties (i.e. there is no $i'_1$ s.t. $i < i'_1 < i_1$ enjoying them),
and for the since operator $i_n$ must be the rightmost.
The latter is a since operator despite being a future modality.
We chose this naming because in formulas such as $\lhysince{\varphi}{\psi}$,
$\psi$ must hold at the beginning of the path, while $\varphi$ must hold in subsequent
positions, which is the typical behavior of since operators.
Note that these paths only contain forward non-maximal chain ends,
which are in the $\lessdot$ relation with $i$.
For example, in pos.~2 $\lhyuntil{\mathbf{throw}}{\mathrm{t}_3}$
holds because of path 6-7-8, since $\mathbf{throw}$ holds in 6-7 and $\mathrm{t}_3$ in 8;
instead, the path made only of pos.~6 and the one made of pos.~6-7,
in which $\mathrm{t}_3$ does not hold, do not satisfy it.
Similarly, $\lhysince{\mathbf{throw}}{\mathrm{t}_1}$
is satisfied by path 6-7-8, but not by the path made of pos.~8
and the one made of pos.~7-8, in which $\mathrm{t}_1$ does not hold.
Position 9 is not included in these paths,
because it does not satisfy the condition $2 \lessdot 9$
(indeed, $2 \gtrdot 9$).

Conversely, \emph{hierarchical take-precedence} until and since operators
($\lhtuntil{}{}$ and $\lhtsince{}{}$)
consider non-maximal chains ending in the current position $j \in U$.
The paths they are based on are sequences of word positions $i_1 < i_2 < \dots < i_n$,
with $i_n < j$, such that for any $1 \leq k \leq n$ we have $i_k \gtrdot j$ and $\chain(i_k,j)$,
and, additionally, there exists no position $i'_k$
that satisfies these two properties and $i_k < i'_k < i_{k+1}$.
For the until operator, $i_1$ must be the leftmost position enjoying these properties,
and for the since operator $i_n$ must be the rightmost
(i.e.\ there is no $i'_n$, $i_n < i'_n < j$, that satisfies them).
Note that $\lhtuntil{}{}$ is an until operator despite being a past modality:
again, the reason is that $\lhtuntil{\varphi}{\psi}$ enforces $\psi$ at the end
of the path, and $\varphi$ in previous positions, making it more similar to an until operator.
In pos.~6, $\lhtuntil{\mathbf{call}}{\mathrm{p}_c}$
is satisfied by path 3-4 and not by the one made only of pos.~3,
because $\mathrm{p}_c$ does not hold in 3, but it does in 4,
and $\mathbf{call}$ holds in 3.
Formula $\lhtsince{\mathbf{call}}{\mathrm{p}_b}$
is satisfied by path 3-4, and not by the one made only of 4,
because $\mathrm{p}_b$ holds in 3 and $\mathbf{call}$ in 4.

\subsection{Equivalence between Operators}
\label{subsec:optl-equivalence}

In this section, we show that hierarchical operators do not contribute to the expressiveness
of OPTL, because OP-summary until operators can be used in place of them,
at the expense of an exponential blowup in formula length.
First, we define two auxiliary formulas:
\[
\sigma_a \equiv
\bigg( \bigwedge_{\ell \in a} \ell \bigg)
\land
\bigg( \bigwedge_{\ell \in (AP \setminus a)} \neg \ell \bigg),
\qquad
\xi_{a \lessdot} \equiv
\bigvee_{b \in \powset{AP} \mid a \lessdot b} \sigma_b.
\]
For any $a \in \powset{AP}$, $\sigma_a$ holds only in positions labeled with $a$.
$\xi_{a \lessdot}$ holds only in positions labeled with a set $b \in \powset{AP}$
such that $a \lessdot b$, and allows to express OP relations in \optl{} formulas.
Formula $\xi_{\gtrdot a}$ is defined similarly.
The yield-precedence hierarchical until operator $\lhyuntil{}{}$ can be translated as follows:
\[
\lhyuntil{\varphi}{\psi} \equiv
\bigvee_{a \in \powset{AP}}
  \Big[
    \sigma_a
    \land
    \lnext \Big(
    \luntil{\gtrdot \doteq}
           {\big(
             \laback \sigma_a \implies \varphi
            \big)}
           {\big(
             \laback \sigma_a
             \land
             \xi_{a \lessdot}
             \land
             \psi
            \big)}
    \Big)
  \Big].
\]
Let us say the formula is evaluated in position $i \in U$, labeled with $a \in \powset{AP}$.
The $\lhyuntil{}{}$ operator considers paths made of positions $k$ such that $\chain(i, k)$
and $i \lessdot k$, i.e. ends of non-maximal chains starting in $i$.
$\varphi$ must hold in all positions of a path except the last one, $j$, in which $\psi$ must hold.
This translation works by letting an OP-summary path run from $i + 1$ to a position $h = j$.
Formula $\laback \sigma_a$ makes sure $h$ is the end of a chain starting in a position
labeled with $a$, which must be $i$.
Moreover, $\xi_{a \lessdot}$ implies $i \lessdot h$,
and $\psi$ must hold in $h$. Therefore, we have $h = j$.
In all positions $i < k < j$ such that $\chain(i, k)$, $\varphi$ must hold:
this is enforced similarly by subformula $\laback \sigma_a \implies \varphi$.
The fact that $\chain(i, j)$ implies $i \lessdot k$ for all positions $k$ above,
so there is no need to put $\xi_{a \lessdot}$ in the left side of the until operator.
Since the OP relations between two positions depend on the subsets of $AP$ labeling them,
the disjunction of separate instances of the formula for each $a \in \powset{AP}$ are needed.

Formulas translating other hierarchical operators are analogous.
Notice that the size of the translation above is exponential in the length of the initial formula.
In fact, let $n = |AP|$:
$|\sigma_a| = \Theta(n)$ and $|\xi_{a \lessdot}| = O(n 2^{n})$.
If $\beta$ is the translation of $\alpha = \lhyuntil{\varphi}{\psi}$,
then in the worst-case scenario we have
 $|\beta| = O((2^n)^{|\alpha|} |\xi_{a \lessdot}|) = O(n 2^{2 n |\alpha|})$.
 
The above result allows us to state that the set of operators
\(
\{
  \neg, \allowbreak
  \land, \allowbreak
  \lnext, \allowbreak
  \lback, \allowbreak
  \lanext, \allowbreak
  \laback, \allowbreak
  \lluntil{}{}, \allowbreak
  \llsince{}{}, \allowbreak
  \luntil{\boxdot}{}{},
  \lsince{\boxdot}{}{}
\}
\)
is adequate for \optl{}.
We conjecture that a result similar to the above could be obtained for the linear until
and since operators as well, allowing us to omit them from this set.

\section{Relationship with Nested Words}
\label{sec:OPTL vs NWTL}

We now explore the relationship between \optl{} and the \nwtl{} logic
presented in \cite{lmcs/AlurABEIL08}.
\nwtl{} is a temporal logic based on the VPL family, which is strictly contained in OPL.
In \cite{CrespiMandrioli12, LonatiEtAl2015, MP18} the relations between the two families are discussed in depth both from a mathematical and an application point of view: building  \optl{} formulas describing OPLs that are not VPLs is a trivial job.
This proves that there exist languages not expressible in \nwtl{} that
can be expressed in \optl{}.
To prove that \optl{} is more expressive than \nwtl{},
we first show a way to translate a nested word into an ``almost isomorphic'' \optl{} structure; then, we give a translation schema for \nwtl{} formulas into equivalent \optl{} ones.

Recall that a nested word is a structure
$NW = \langle U, (P_\mathrm{a})_{\mathrm{a} \in \Lambda}, <, \munw, \mathtt{call}, \mathtt{ret} \rangle$,
where $\Lambda$ is the set of atomic propositions;
$U$ is a set of word positions such that $U = \{1, \dots, n\}$ if $NW$ is finite,
and $U = \mathbb{N}$ if it is a nested $\omega$-word;
`$<$' is the ordering of $\mathbb{N}$;
$P_\mathrm{a}$ is the set of positions labeled with $\mathrm{a}$.
$\munw$ is a binary relation and $\mathtt{call}$ and $\mathtt{ret}$ are unary relations
such that $\munw(i, j)$ implies $\mathtt{call}(i)$ and $\mathtt{ret}(j)$;
if $\munw(i,j)$ and $\munw(i,j')$ hold then $j=j'$,
while if $\munw(i,j)$ and $\munw(i',j)$ then $i=i'$;
and if $i \leq j$ and $\mathtt{call}(i)$ and $\mathtt{ret}(j)$ then there exists a position $k$
such that $i \leq k \leq j$, and either $\munw(i,k)$ or $\munw(k,j)$.
If $\munw(i,j)$ then $i$ is the matching call of $j$, which is the matching return of $i$.
If $\mathtt{call}(i)$ (resp. $\mathtt{ret}(j)$) but for no $j \in U$ (resp. $i \in U$)
we have $\munw(i,j)$, then $i$ (resp. $j$) is a \emph{pending} call (resp. return).

The syntax of NWTL is given by
\(
  \varphi :=
  \top 
  \mid \mathrm{a}
  \mid \mathtt{call}
  \mid \mathtt{ret}
  \mid \neg \varphi
  \mid \varphi \lor \varphi
  \mid \lnext \varphi
  \mid \lnext_\munw \varphi
  \mid \lback \varphi
  \mid \lback_\munw \varphi
  \mid \luntil{\sigma}{\varphi}{\varphi}
  \mid \lsince{\sigma}{\varphi}{\varphi}
\),
with $\mathrm{a} \in \Lambda$.
The semantics of propositional and LTL-like operators is the usual one.
For any $i \in U$,
$(NW, i) \models \mathtt{call}$ iff $\mathtt{call}(i)$ and similarly for $\mathtt{ret}$;
$(NW, i) \models \lnext_\munw \varphi$ iff there exists $j \in U$ such that $\munw(i,j)$
and $(NW, j) \models \varphi$, while the meaning of $\lback_\munw \varphi$ is analogous,
but it refers to the past.
The until and since operators $\luntil{\sigma}{}{}$ and $\lsince{\sigma}{}{}$ are based on summary paths.
A summary path between $i, j \in U$, $i < j$, is a sequence of positions
$i = i_0 < i_1 < \dots < i_n = j$ such that for any $0 \leq k < n$ we have either
$i_{k+1} = h$ if $\munw(i_k,h)$ and $h \leq j$, or
$i_{k+1} = i_k + 1$ otherwise.
  
Given any nested word $NW$ as defined above,
it is possible to build an equivalent algebraic structure for \optl{} as
$OW = \langle U', M^{NW}, P' \rangle$.
Given $U = \{ 1, \dots, n \}$, we have $U' = U \cup \{ 0, n+1 \}$.
The set of propositional letters is $AP = \Lambda \cup \Sigma$
with $\Sigma = \{ \mathbf{call}, \mathbf{ret}, \mathbf{i} \}$.
For any $i \in U$ we define
$P'(i) = \{ \mathrm{a} \in \Lambda \mid i \in P_\mathrm{a} \} \cup \sigma(i)$,
where $\sigma(i) = \{\mathbf{call}\}$ iff $\mathtt{call}(i)$,
$\sigma(i) = \{\mathbf{ret}\}$ iff $\mathtt{ret}(i)$,
and $\sigma(i) = \{\mathbf{i}\}$ otherwise.
Finally, the OPM $M^{NW}$ is shown in Figure~\ref{fig:nw-translation}.

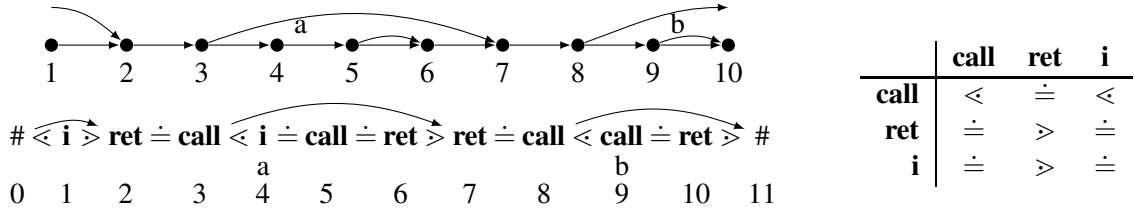
\begin{figure}
\begin{subfigure}{0.69\linewidth}
\begin{subfigure}{\linewidth}
\centering
\begin{tikzpicture}
[dot/.style={circle, fill=black, inner sep=0pt, minimum size=5pt},
edge/.style={->, >=latex}]
  \node[dot] (n1) at (0,0) [label=below:1] {};
  \node[dot] (n2) at (1,0) [label=below:2] {};
  \node[dot] (n3) at (2,0) [label=below:3] {};
  \node[dot] (n4) at (3,0) [label=below:4, label=10:$\mathrm{a}$] {};
  \node[dot] (n5) at (4,0) [label=below:5] {};
  \node[dot] (n6) at (5,0) [label=below:6] {};
  \node[dot] (n7) at (6,0) [label=below:7] {};
  \node[dot] (n8) at (7,0) [label=below:8] {};
  \node[dot] (n9) at (8,0) [label=below:9, label=10:$\mathrm{b}$] {};
  \node[dot] (n10) at (9,0) [label=below:10] {};

  \draw[edge, >=latex] (n1) -- (n2);
  \draw[edge] (n2) -- (n3);
  \draw[edge] (n3) -- (n4);
  \draw[edge] (n4) -- (n5);
  \draw[edge] (n5) -- (n6);
  \draw[edge] (n6) -- (n7);
  \draw[edge] (n7) -- (n8);
  \draw[edge] (n8) -- (n9);
  \draw[edge] (n9) -- (n10);

  \draw[edge] (0,0.5) to [out=0, in=140] (n2);
  \draw[edge] (n3) to [out=20, in=160] (n7);
  \draw[edge] (n5) to [out=20, in=160] (n6);
  \draw[edge] (n8) to [out=20, in=180] (9,0.5);
  \draw[edge] (n9) to [out=20, in=160] (n10);
\end{tikzpicture}
\end{subfigure}
\begin{subfigure}{\linewidth}
\centering
\begin{tikzpicture}
[scale=0.8, edge/.style={->, >=latex}]
\matrix (m) [matrix of math nodes, column sep=-5pt, row sep=-4pt]
{
  \bfsym{\#}
  & \lessdot & \mathbf{i} & \gtrdot & \mathbf{ret} & \doteq & \mathbf{call}
  & \lessdot & \mathbf{i} & \doteq & \mathbf{call} & \doteq & \mathbf{ret}
  & \gtrdot & \mathbf{ret} & \doteq & \mathbf{call} & \lessdot & \mathbf{call}
  & \doteq & \mathbf{ret} & \gtrdot
  & \bfsym{\#} \\
  & & & & & & & & \mathrm{a} & & & & & & & & & & \mathrm{b} & & & & \\
  0 & & 1 & & 2 & & 3 & & 4 & & 5 & & 6 & & 7 & & 8 & & 9 & & 10 & & 11 \\
};
  \draw[edge] (m-1-1) to [out=20, in=160] (m-1-5);
  \draw[edge] (m-1-7) to [out=20, in=160] (m-1-15);
  \draw[edge] (m-1-17) to [out=20, in=160] (m-1-23);
\end{tikzpicture}
\end{subfigure}
\end{subfigure}
\begin{subfigure}{0.3\linewidth}
\centering
\begin{tabular}{r | c c c }
                & $\mathbf{call}$ & $\mathbf{ret}$ & $\mathbf{i}$  \\
\hline
$\mathbf{call}$ & $\lessdot$      & $\doteq$       & $\lessdot$    \\
$\mathbf{ret}$  & $\doteq$        & $\gtrdot$      & $\doteq$      \\
$\mathbf{i}$    & $\doteq$        & $\gtrdot$      & $\doteq$      \\
\end{tabular}
\end{subfigure}
\caption{The top-left figure is the representation of a nested word example,
  and its translation into an \optl{} structure is shown below,
  w.r.t.\ the OPM $M^{NW}$ on the right.}
\label{fig:nw-translation}
\end{figure}

An example nested word is shown in Figure~\ref{fig:nw-translation},
along with its translation into an \optl{} word.
In the translation, all the call positions form the context of a chain with the matched return, except for consecutive positions $i, i+1 \in U$ such that $i$ is a call and $i+1$ a return.
Therefore, we are able to use the chain relation
to translate the matching relation of nested words,
except for consecutive call/return positions, which have to be considered separately.
Also, unmatched returns and calls form a chain with the first and last $\bfsym{\#}$ positions.
This is formalized by the following properties, whose proofs are omitted:
\begin{enumerate}
\item
  For any two distinct positions $i, j \in U$,
  $i < j$, if $\chain(i,j)$ holds
  then $\mathbf{call} \in P'(i)$ and $\mathbf{ret} \in P'(j)$.
\item \label{item:nw-chain-one-to-one}
  For any $i, j, j' \in U$ if $\chain(i, j)$ and $\chain(i, j')$ then $j = j'$,
  and for any $i, i', j \in U$ if $\chain(i, j)$ and $\chain(i', j)$ then $i = i'$.
\item \label{item:mu-nw-iff-chain-ow}
  For any $i, j \in U$ such that $j > i + 1$, we have $\munw(i, j)$ iff $\chain(i, j)$.
\item \label{item:summary-paths-opsp}
  Given any two positions $i, j \in U$, $i \leq j$,
  the summary path between $i$ and $j$ in $NW$ coincides with the
  OP-summary path between the same positions based on OP relations
  $O = \{ \lessdot, \allowbreak \doteq, \allowbreak \gtrdot \}$ in $OW$.
\end{enumerate}

After establishing a certain degree of isomorphism between nested words
and their \optl{} translations, we can give a translation schema from \nwtl{} to \optl{} formulas.

\begin{theorem}[\nwtl{} $\subseteq$ \optl{}]
  \label{thm:nwtl-subseteq-optl}
  Given an \nwtl{} formula $\varphi$, it is possible to translate it to
  an \optl{} formula $\varphi'$ of length linear in $|\varphi|$ such that,
  for any nested word $w$ and position $i$,
  if $w$ is translated into an OP word $w'$ as described at the beginning of
  Section~\ref{sec:OPTL vs NWTL},
  then $(w, i) \models \varphi$ iff $(w', i) \models \varphi'$, with $i \in U$.
\end{theorem}
\begin{proof}
  Let $w'$ be an OP word built from $w$ as described above.
  For any \nwtl{} formula $\varphi$ we denote as $\varphi' = \alpha(\varphi)$
  the \optl{} formula that satisfies
  $(w, i) \models \varphi$ iff $(w', i) \models \varphi'$.
  The translation function $\alpha$ is defined for non-trivial arguments as follows:
  
\noindent $\bullet$
\(\alpha(\lnext_\munw \varphi)
    = (\luntil{\doteq}{\mathbf{call}}{(\mathbf{ret} \land \alpha(\varphi))})
    \land \neg \mathbf{ret}
    \),
    and \(\alpha(\lback_\munw \varphi)
    = (\lsince{\doteq}{\mathbf{ret}}{(\mathbf{call} \land \alpha(\varphi))})
    \land \neg \mathbf{call}
    \).
    \footnote{We chose these translations instead of the more straightforward
    \(
      \alpha(\lnext_\munw \varphi) =
      \lanext \alpha(\varphi)
      \lor (\mathbf{call} \land \lnext(\mathbf{ret} \land \alpha(\varphi))
    \)
    (and an analogous one for $\lback_\munw \varphi$)
    because the latter causes an exponential blowup in formula length.}

    Let
    $\gamma \equiv \luntil{\doteq}{\mathbf{call}}{(\mathbf{ret} \land \alpha(\varphi))}$.
    In \nwtl{} we have $(w, i) \models \lnext_\munw \varphi$
    iff there exists $j \in U$ such that $\munw(i, j)$ and $(w, j) \models \varphi$.
    Because of Prop.~\ref{item:mu-nw-iff-chain-ow}, if such a position $j$ exists
    and $j > i+1$ then also $\chain(i, j)$ holds and,
    because of Prop.~\ref{item:nw-chain-one-to-one}, we have $\fchain(i, j)$.
    Consider the path only made of $i$ and $j$:
    it is an OP summary path, because it falls in the first case of the definition.
    Since by construction $\mathbf{call} \in P'(i)$ and $\mathbf{ret} \in P'(j)$,
    if $\alpha(\varphi)$ holds in $j$, then $\gamma$ is satisfied,
    and this is the only path in which $\gamma$ is true.
    Paths terminating in a position strictly between $i$ and $j$
    are forbidden by allowing only the $\doteq$ relation;
    also, all paths surpassing $j$ must include it,
    but $\mathbf{call} \notin P'(j)$ falsifies $\gamma$.
    If $i$ and $j$ are such that $\munw(i, j)$ but $j = i+1$,
    we have $\mathbf{call} \in P'(i)$ and $\mathbf{ret} \in P'(j)$,
    so $i \doteq j$, and the path made of $i$ and $j$ is valid.
    Finally, if $(w, i) \not\models \lnext_\munw \varphi$
    because position $i$ is not a matched call,
    then $\mathbf{i} \in P'(i)$ and $\gamma$ is false because
    $\mathbf{call}, \mathbf{ret} \notin P'(i)$.
    If $i$ is an unmatched call, then $\chain(i, n+1)$, but $\mathbf{ret} \notin P'(n+1)$,
    which falsifies $\gamma$.
    If $\mathtt{ret}(i)$, then $\neg \mathbf{ret}$ is false.
    The translation for $\lback_\munw$ is similar.
    
\noindent $\bullet$
  \(\alpha(\luntil{\sigma}{\varphi}{\psi})
    = \luntil{\gtrdot \doteq \lessdot}{\alpha(\varphi)}{\alpha(\psi)}\)
    and \(\alpha(\lsince{\sigma}{\varphi}{\psi})
    = \lsince{\gtrdot \doteq \lessdot}{\alpha(\varphi)}{\alpha(\psi)}\):
    from Prop.~\ref{item:summary-paths-opsp} we know that
    the set of summary paths starting from position $i$ in $w$ corresponds to the
    set of OP summary paths starting from $i$ in $w'$,
    which implies the OP summary operators coincide with their \nwtl{} counterparts.
  
By induction on the syntactic structure of $\varphi$,
we can conclude \nwtl{} $\subseteq$ \optl{}.
\end{proof}

For example, take formula $\varphi = \luntil{\sigma}{(\neg \mathrm{a})}{\mathrm{b}}$:
the nested word of Fig.~\ref{fig:nw-translation} satisfies $\varphi$
because of summary path 1-2-3-7-8-9.
$\varphi$ is translated into
\(
  \varphi' =
    \luntil{\gtrdot \doteq \lessdot}
           {(\neg \mathrm{a})}
           {\mathrm{b}},
\)
which is satisfied by the \optl{} structure of Fig.~\ref{fig:nw-translation},
where the OP summary path covering the same positions above witnesses its truth.

\section{Model Checking Operator Precedence Languages}
\label{sec:model checking}

Given a set of atomic propositions $AP$, an OPM $M_\powset{AP}$, and an \optl{} formula $\varphi$ containing atomic propositions in $AP$,
we explain how to build an automaton
$\mathcal{A}_\varphi$ accepting the finite strings on $\powset{AP}$ satisfying $\varphi$, thus compatible with $M_\powset{AP}$.
We later sketch the way to extend this definition to infinite words.
Let
\(
  \mathcal{A}_\varphi =
  \langle
     \powset{AP}, \allowbreak
     M_\powset{AP}, \allowbreak
     \atoms{\varphi}, \allowbreak
     I, \allowbreak
     F, \allowbreak
     \delta
  \rangle
\)
be an OPA, with $M_\powset{AP}$ being an OPM,
and $\atoms{\varphi}$ being the set of states.
We first describe how the OPA is built for \ltl{}-like operators,
which is done with the classic procedure by \cite{WolperVardiSistla1983},
and then treat other operators separately.
Initially, the \emph{closure} $\clos{\varphi}$
of a formula $\varphi$ is defined as the smallest set such that $\varphi \in \clos{\varphi}$,
$AP \subseteq \clos{\varphi}$, if $\psi \in \clos{\varphi}$ and $\psi \neq \neg \theta$
for any \optl{} formula $\theta$, then $\neg \psi \in \clos{\varphi}$;
if any of $\neg \psi$, $\lnext \psi$ or $\lback \psi$ is in $\clos{\varphi}$,
then $\psi \in \clos{\varphi}$;
if any of $\psi \land \theta$ or $\psi \lor \theta$ is in $\clos{\varphi}$,
then $\psi \in \clos{\varphi}$ and $\theta \in \clos{\varphi}$.
We will further enrich $\clos{\varphi}$ when dealing with non-\ltl{} operators.
The set $\atoms{\varphi}$ contains all consistent sets $\Phi \subseteq \clos{\varphi}$,
i.e., such that for every $\psi \in \clos{\varphi}$, $\psi \in \Phi$ iff $\neg \psi \notin \Phi$;
if $\psi \land \theta \in \Phi$, then $\psi \in \Phi$ and $\theta \in \Phi$;
if $\psi \lor \theta \in \Phi$, then $\psi \in \Phi$ or $\theta \in \Phi$.
The set of initial states $I$ consists only of atoms containing $\varphi$,
but with no $\lback$ formula.
The set of final states $F$ contains all atoms $\Phi \in \atoms{\varphi}$
such that for any $\lnext \psi \in \clos{\varphi}$, $\lnext \psi \not\in \Phi$,
so no word terminating with unsatisfied temporal requirements can be accepted,
and $AP \cap \Phi = \{ \bfsym{\#} \}$, so the last position of each word must contain
the terminal symbol only.
Satisfaction of temporal constraints is achieved by the transition relation $\delta$.
For any $\Phi, \Theta \in \atoms{\varphi}$ and $a \in \powset{AP}$,
$(\Phi, a, \Theta) \in \delta_\text{push} \cap \delta_\text{shift}$
iff for any $\mathrm{p} \in AP$, $\mathrm{p} \in \Phi$ iff $\mathrm{p} \in a$;
$\lnext \psi \in \Phi$ iff $\psi \in \Theta$;
$\lback \psi \in \Theta$ iff $\psi \in \Phi$, for any formula $\psi$.
Since only push and shift transitions read terminal symbols, they fulfill next and back requirements.
Pop transitions are important for checking operators depending on the chain relation,
but for the moment we prevent them from removing sub-formulas introduced in a previous state
by other transitions: for any $\Phi, \Theta, \Psi \in \atoms{\varphi}$,
$(\Phi, \Theta, \Psi) \in \delta_\mathit{pop}$ only if $\Psi$ is the smallest set such that
$\Phi \subseteq \Psi$, and for any $\mathrm{p} \in AP$, $\mathrm{p} \in \Phi$
iff $\mathrm{p} \in \Psi$, and containing additional formulas required by the rules detailed in the next paragraphs.

Additionally, for the linear until operator, let $\Phi \in \atoms{\varphi}$,
with $\lluntil{\psi}{\theta} \in \Phi$.
Then $\psi, \theta, \lnext(\lluntil{\psi}{\theta}) \in \clos{\varphi}$,
and either $\theta \in \Phi$,
or $\psi \in \theta$ and $\lnext(\lluntil{\psi}{\theta}) \in \Phi$.

\begin{figure}
\begin{subfigure}{0.3\linewidth}
\centering
\begin{tabular}{r | c c c}
  & $\mathbf{a}$ & $\mathbf{b}$ & $\mathbf{c}$ \\
  \hline
  $\mathbf{a}$ & $\gtrdot$ & $\lessdot$ & $\doteq$ \\
  $\mathbf{b}$ & $\gtrdot$ & $\gtrdot$  & $\gtrdot$ \\
  $\mathbf{c}$ & $\gtrdot$ & $\gtrdot$  & $\lessdot$ \\
\end{tabular}
\end{subfigure}
\begin{subfigure}{0.7\linewidth}
\centering
{\footnotesize
\begin{tabular}{c | r | r | r}
  step & input & state & stack \\
  \hline
  1 & $\mathbf{a} \lessdot \mathbf{b} \gtrdot \mathbf{b} \gtrdot \mathbf{c} \gtrdot \bfsym{\#}$ & $\{\mathbf{a}, \lanext \mathbf{c}\}$
    & $\bot$ \\
  2 & $\mathbf{b} \gtrdot \mathbf{b} \gtrdot \mathbf{c} \gtrdot \bfsym{\#}$
    & $\{\mathbf{b}, \lapnext \mathbf{c}\}$
    & $[\mathbf{a}, \{\mathbf{a}, \lanext \mathbf{c}\}] \bot$ \\
  3 & $\mathbf{b} \gtrdot \mathbf{c} \gtrdot \bfsym{\#}$ & $\{\mathbf{b}\}$
    & $[\mathbf{b}, \{\mathbf{b}, \lapnext \mathbf{c}\}] [\mathbf{a}, \{\mathbf{a}, \lanext \mathbf{c}\}] \bot$ \\
  4 & $\mathbf{b} \gtrdot \mathbf{c} \gtrdot \bfsym{\#}$ & $\{\mathbf{b}, \lapnext \mathbf{c}\}$
    & $[\mathbf{a}, \{\mathbf{a}, \lanext \mathbf{c}\}] \bot$ \\
  5 & $\mathbf{c} \gtrdot \bfsym{\#}$ & $\{\mathbf{c}\}$
    & $[\mathbf{b}, \{\mathbf{b}, \lapnext \mathbf{c}\}] [\mathbf{a}, \{\mathbf{a}, \lanext \mathbf{c}\}] \bot$ \\
  6 & $\mathbf{c} \gtrdot \bfsym{\#}$ & $\{\mathbf{c}, \lapnext \mathbf{c}, \laenext \mathbf{c}\}$
    & $[\mathbf{a}, \{\mathbf{a}, \lanext \mathbf{c}\}] \bot$ \\
  7 & $\bfsym{\#}$ & $\{\bfsym{\#}\}$
    & $[\mathbf{c}, \{\mathbf{a}, \lanext \mathbf{c}\}] \bot$ \\
  8 & $\bfsym{\#}$ & $\{\bfsym{\#}\}$
    & $\bot$ \\
\end{tabular}
}
\end{subfigure}
\caption{OPM (left) and a computation (right) on $\mathbf{a} \mathbf{b} \mathbf{b} \mathbf{c}$ of the automaton for formula $\lanext \mathbf{c}$.
 }
\label{fig:lanext-plus-accepting-run}
\end{figure}

\noindent \textbf{Matching Next ($\lanext$) Operator.}
The $\lanext$ and $\laback$ operators are defined w.r.t.\ maximal chains:
the main difficulty posed by their model-checking is building
automata able to discern them from inner chains, when multiple chains start in the same position.
The automaton for model-checking $\lanext$ relies on the stack
in order to keep track of its requirements:
if a formula $\lanext \psi$ holds in a state where a chain begins,
the automaton stores the auxiliary symbol $\lapnext \psi$ onto the stack,
and pops it when the chain ends, forcing $\psi$ to hold in the appropriate state.
More formally, if $\lanext \psi \in \clos{\varphi}$, then we add
\(
  \psi,
  \lapnext \psi,
  \laenext \psi
  \in \clos{\varphi}
\).
Also, we enrich the previous constraints on $\delta$ as follows:
for any $\Phi, \Theta, \Psi \in \atoms{\varphi}$ and $a \in \powset{AP}$,
\begin{inparaenum}[(1)]
\item \label{item:lanext-push-shift-constr}
  $(\Phi, a, \Theta) \in \delta_\text{push} \cap \delta_\text{shift}$
  iff when $\lanext \psi \in \Phi$, we have $\lapnext \psi \in \Theta$
  and $\laenext \psi \not\in \Theta$;
\item \label{item:lanext-end-shift}
  $(\Phi, a, \Theta) \in \delta_\text{shift}$
  iff whenever $\lapnext \psi \in \Phi$, we also have $\psi \in \Phi$,
  $\laenext \psi \in \Phi$, and $\laenext \psi \not\in \Theta$;
  $(\Phi, \Theta, \Psi) \in \delta_\mathit{pop}$ iff
\item \label{item:if-lapnext-theta-lapnext-psi}
  when $\lapnext \psi \in \Theta$, we have $\lapnext \psi \in \Psi$, and
\item \label{item:if-lapnext-phi-lapnext-psi}
  when $\lapnext \psi \in \Phi$, then also $\psi \in \Phi$,
  $\laenext \psi \in \Phi$, and $\laenext \psi \not\in \Psi$.
\end{inparaenum}
Finally, for any formula $\lanext \psi \in \clos{\varphi}$,
we must exclude from $F$ all atoms containing $\lapnext \psi$
or $\lanext \psi$.

To explain how the rules above work, we refer to the example computation
of the automaton for formula $\lanext \mathbf{c}$
of Fig.~\ref{fig:lanext-plus-accepting-run}.
$\lanext \mathbf{c}$ holds in the initial state:
since the symbol $\bfsym{\#}$ yields precedence to all other symbols,
the first one (namely $a = \{\mathbf{a}\}$) is consumed by a push transition.
Due to constraint (\ref{item:lanext-push-shift-constr}),
the auxiliary operator $\lapnext \psi$ is forced into the next state.
If $a$ was not the first symbol, and it had been read by a shift,
the same constraint would hold.
Since a chain starts in $a$, the next symbol is also read by a push transition,
which stores the previous state, containing $\lapnext \mathbf{c}$,
on the stack (step 2-3 of Fig.~\ref{fig:lanext-plus-accepting-run}).
Then, the chain ends with the second $\{\mathbf{b}\}$ symbol,
and the state mentioned above is popped.
The operand of $\lanext$ must hold at the end of the outermost chain,
and this is an inner one. However, the automaton still cannot discern between the two cases,
and must wait for the next transition to occur,
so rule (\ref{item:if-lapnext-theta-lapnext-psi}) just forces $\lapnext \mathbf{c}$ into
the next state (step 3-4).
If the next transition is a push, it means $a$ yields precedence to the symbol read by it,
and this is not the outermost chain: the current state, containing $\lapnext \mathbf{c}$,
is again pushed on stack (step 4-5).
The process goes on until the pop transition at the end of a chain starting in $a$
is followed by another pop or a shift: this either means $a$ takes precedence
from the current symbol, and its stack symbol must be popped,
or they are equal in precedence. In both cases, we reached the end of the outermost chain,
and $\psi$ must hold in the next state.
First, in step 5-6 rule (\ref{item:if-lapnext-theta-lapnext-psi})
puts $\lapnext \mathbf{c}$ in the current state:
for rule (\ref{item:lanext-end-shift}) to be satisfied by the shift transition of step 6-7,
also $\mathbf{c}$ and $\laenext \mathbf{c}$ must hold.
If in step 6-7 a pop transition occurred instead of a shift,
the same constraints would have been enforced by rule (\ref{item:if-lapnext-phi-lapnext-psi}).
The symbol $\laenext$ is another auxiliary operator that marks the satisfaction
of the original $\lanext$ formula, and its purpose is to make rule
(\ref{item:lanext-push-shift-constr})
block computations in which a chain does not start in $a$,
but $\psi$ holds in the next symbol.
The $\laback$ operator can be treated in an analogous way.

\noindent \textbf{OP Summary Until (Since).}
For any $O \subseteq \{ \lessdot, \doteq, \gtrdot \}$,
if $\luntil{O}{\psi}{\theta} \in \clos{\varphi}$ then
\(
  \psi, \allowbreak
  \theta, \allowbreak
  \lanext(\luntil{O}{\psi}{\theta}), \allowbreak
  \lnext(\luntil{O}{\psi}{\theta})
  \in \clos{\varphi}
\).
Also, for any $\Phi \in \atoms{\varphi}$,
$\luntil{O}{\psi}{\theta} \in \Phi$ iff either:
\begin{inparaenum}[(1)]
\item \label{item:opsuntil-trivial}
  $\theta \in \Phi$; or
\item \label{item:opsuntil-lanext}
  $\psi \in \Phi$ and $\lanext(\luntil{O}{\psi}{\theta}) \in \Phi$; or
\item \label{item:opsuntil-lnext}
  $\psi \in \Phi$ and $\lnext(\luntil{O}{\psi}{\theta}) \in \Phi$.
\end{inparaenum}
If (\ref{item:opsuntil-trivial}) holds, then $\luntil{O}{\psi}{\theta}$ is trivially true;
if (\ref{item:opsuntil-lanext}) holds, the path skips the body of a chain
starting in the current position, where $\psi$ holds;
if (\ref{item:opsuntil-lnext}) holds, then $\psi$ is true in the current position
and the path continues in the next one.
In the latter case, we must make sure the path is followed only if the current position and
the next one are in one of the precedence relations in $O$.
This can be achieved by adding the following constraints:
if only (\ref{item:opsuntil-lnext}) holds in a state $\Phi$ and $\luntil{O}{\psi}{\theta} \in \Phi$,
then for any $a, b \in \powset{AP}$ and $\Theta \in \atoms{\varphi}$,
with $b = \Theta \cap AP$, we have
$(\Phi, a, \Theta) \in \delta_\mathit{push} \cup \delta_\mathit{shift}$
only if $a \odot b$ and $\odot \in O$.
We need not impose constraints on pop transitions because they do not consume input symbols,
and they preserve the same subset of $AP$ contained in the starting state.
Furthermore, for any $\Phi \in \atoms{\varphi}$,
if $\Phi \in F$ then $\luntil{O}{\psi}{\theta} \not\in \Phi$,
unless $\theta \in \Phi$.

The construction for the since counterpart is analogous:
$\lanext$ and $\lnext$ are substituted with $\laback$ and $\lback$.

\noindent \textbf{Yield-Precedence Hierarchical Until.}
This construction exploits the automaton's stack to keep track of the requirements
of the $\lhyuntil{}{}$ operator in a way similar to $\lanext$.
If $\lhyuntil{\psi}{\theta} \in \clos{\varphi}$
we introduce the auxiliary operators $\lhypuntil{}{}$ and $\lhyeuntil{}{}$.
We add
\(
  \lhypuntil{\psi}{\theta}, \allowbreak
  \lnext (\lhypuntil{\psi}{\theta}), \allowbreak
  \lhyeuntil{\psi}{\theta}
  \in \clos{\varphi}
\),
and, for any $\Phi \in \atoms{\varphi}$
such that $\lhyuntil{\psi}{\theta} \in \Phi$, also $\lnext (\lhypuntil{\psi}{\theta}) \in \Phi$.
Given any $\Phi, \Theta, \Psi \in \atoms{\varphi}$,
we have $(\Phi, \Theta, \Psi) \in \delta_\mathit{pop}$
\begin{inparaenum}[(1)]
\item \label{item:yphu-pop-stop-maximal}
  iff $\lhypuntil{\psi}{\theta} \not\in \Phi$ and
  $\lhyeuntil{\psi}{\theta} \not\in \Phi$,
  and if $\lhypuntil{\psi}{\theta} \in \Theta$, then either
\item \label{item:yphu-pop-end}
  $\theta \in \Psi$ and $\lhyeuntil{\psi}{\theta} \in \Psi$, or
\item
  $\psi \in \Psi$ and $\lhypuntil{\psi}{\theta} \in \Psi$.
Moreover, for any $\Phi, \Theta \in \atoms{\varphi}$ and $a \in \powset{AP}$,
\item \label{item:yphu-shift-stop-maximal}
  if $(\Phi, a, \Theta) \in \delta_\mathit{shift}$ then
  $\lhypuntil{\psi}{\theta}, \lhyeuntil{\psi}{\theta} \not\in \Phi$;
  and
\item
  if $(\Phi, a, \Theta) \in \delta_\mathit{push}$ then
  $\lhyeuntil{\psi}{\theta} \not\in \Theta$.
\end{inparaenum}
Moreover, states containing $\lhypuntil{\psi}{\theta}$ or $\lhyeuntil{\psi}{\theta}$
are excluded from the final set $F$.

Similarly to what happens for $\lanext$, the auxiliary operator $\lhypuntil{\psi}{\theta}$
is put into the state after the one in which $\lhyuntil{\psi}{\theta}$ holds,
this time by a $\lnext$ operator.
Then, the state containing the auxiliary symbol is pushed on stack.
Each time it is popped, we check if either $\theta$ holds in the next state
(rule (\ref{item:yphu-pop-end})), and $\lhyeuntil{\psi}{\theta}$ is put into the next state
to mark the end of the hierarchical path, or if $\psi$ holds and the path continues,
so $\lhypuntil{\psi}{\theta}$ must continue being propagated.
Let $j$ be the ending position of the maximal chain starting
where $\lhyuntil{\psi}{\theta}$ is supposed to hold.
The purpose of rules (\ref{item:yphu-pop-stop-maximal}) and
(\ref{item:yphu-shift-stop-maximal}) is to stop the computation
if the hierarchical path does not contain a position before $j$ in which $\theta$ holds
(recall that $j$ is not part of the hierarchical path, only non-maximal chain ends are). 

\noindent \textbf{Other Hierarchical Operators.}
The automata recognizing other hierarchical operators can be built similarly.
In particular, the automaton for $\lhysince{}{}$ is similar to the one for $\lhyuntil{}{}$,
while those for $\lhtuntil{}{}$ and $\lhtsince{}{}$ formulas are closer to the one recognizing
$\laback$.

\noindent \textbf{Complexity.}
The size of set $\clos{\varphi}$ is linear w.r.t the length of formula $\varphi$,
and set $\atoms{\varphi}$, which contains the states of the automaton,
is a subset of $\powset{\clos{\varphi}}$, being of size exponential in $\clos{\varphi}$.
Therefore, we can state that for any \optl{} formula $\varphi$
it is possible to build an OPA of size $2^{O(|\varphi|)}$
that accepts models satisfying $\varphi$.
Note that, despite translating hierarchical operators into summary operators
yields an exponential blow-up of formula length (see Section~\ref{subsec:optl-equivalence}),
they do not lead to an increase in the complexity of model checking.
This is in contrast with the behavior of the within operator in \nwtl{},
which exponentially increases model checking complexity.

\subsection{Hints for Model Checking Infinite Words}
\label{MC infinite}
Model checking for an \optl{} formula $\varphi$ on infinite words
can be performed by building a generalized OP B\"uchi Automaton
($\omega$OPBA, \cite{LonatiEtAl2015})
\(
  \mathcal{A}_\varphi^\omega =
  \langle \powset{AP}, M_\powset{AP}, \atoms{\varphi}, I, \mathbf{F}, \delta \rangle
\),
which differs from the finite-word counterpart only for the acceptance condition.
$\mathbf{F}$ is the set of sets of B\"uchi final states,
and an $\omega$-word is accepted iff at least one state from each one of the
sets contained in $\mathbf{F}$ is visited infinitely often during the computation.
The main difficulty in adapting the finite-word construction to the infinite case
is the acceptance condition.
In finite words, the stack is empty at the end of every accepting computation,
which implies all temporal constraints tracked by stack symbols must have been satisfied.
In $\omega$OPBAs, the stack may never be empty, and symbols containing auxiliary operators
may remain buried forever, never enforcing the satisfaction of the formulas they refer to.
This problem can be solved by defining additional auxiliary operators for each temporal formula
that requires stack support, such as $\lanext$, $\laback$ and hierarchical operators.
Suppose we want to model check $\lanext \psi$.
One of these new symbols, $\laknext \psi$, must be inserted in the current state whenever
$\lapnext \psi$ is pushed on stack, and kept in the automaton's state until $\lapnext \psi$
is popped, and its temporal requirement satisfied.
When this happens another symbol, say $\laoenext \psi$,
must be placed into the current state to mark satisfaction, and removed afterwards.
Then, it is possible to define an acceptance set $F \in \mathbf{F}$
for each use of this temporal operator.
$F$ only contains atoms that either contain $\laoenext \psi$
(so the formula can be satisfied infinitely often),
or atoms not containing $\laknext \psi$
(so words where $\lanext \psi$ needs not be satisfied infinitely often are accepted).
When adding constraints on the transition relation to realize this behavior,
particular care must be taken in order to prevent $\laoenext \psi$ form remaining when it should not,
and not to remove $\laknext \psi$ too early when multiple instances of the same subformula
are pending, and only one of them is satisfied.

\section{Conclusions}
\label{sec:conclusion}

We presented a new temporal logic based on the OPL family.
We proved that it is more expressive than \nwtl{}, which is based on the less powerful class VPL. \optl{} can express more properties than other similar
temporal logics, with the possibility of model checking with an automaton of size not greater
than competing formalisms. A further natural step in the theoretical characterization of \optl{} is a more complete exploration of its expressiveness, in particular concerning First-Order completeness, as it has been done for various temporal logics, including NWTL \cite{lmcs/AlurABEIL08}.

On the side of practical application we are planning the development and implementation of suitable model checking algorithms exploiting the theoretical procedures presented in this paper. We also believe that ``core temporal logic languages", such as LTL, CTL, CTL*, and, even more, those derived therefrom to attack subclasses of CFLs, including \nwtl{} and \optl{}, lack user-friendliness and require excessive mathematical skill to translate an informal requirement into a well-defined formula; thus, we envisage a higher level interface for \optl{} more palatable for users familiar with semi-formal notations like UML.

\bibliographystyle{eptcs}
\bibliography{optl}

\end{document}